\documentclass[envcountsame,runningheads]{llncs}

\usepackage{ALgo}
\usepackage[utf8]{inputenc}
\usepackage{graphicx}

\def\s#1{\textrm{\boldmath $#1$}}

\def\pref#1{\mbox{pref(\s{#1})}}
\def\suff#1{\mbox{suff(\s{#1})}}

\def\:{\mbox{\ :\ }}

\def\+{\!+\!}
\def\-{\!-\!}

\def\pref(#1,#2){$#1$ is a prefix of $#2$}
\def\suff(#1,#2){$#1$ is a suffix of $#2$}
\def\reg(#1,#2){$#2$ is $#1$-regular}
\def\notreg(#1,#2){$#2$ is not $#1$-regular}

\def\alphabet{\Sigma}
\def\dalphabet{\Delta}
\def\Deltasigma{\dalphabet_\alphabet}

\def\rank{\iit{rank}}

\def\pp{\mathinner{\ldotp\ldotp}}

\def\BWT{\textit{BWT}}

\def\H{H}

\def\Step{\textit{OneStep}}
\def\Stepi{\textit{Step}}
\def\Merge{\textit{Merge}}
\def\rank{\textit{rank}}

\begin{document}

\title{Efficient pattern matching in degenerate strings with the Burrows--Wheeler transform}
\titlerunning{Pattern matching in degenerate strings} 

\author{Jacqueline W. Daykin\,$^{1,2,3}$ \and
Richard Groult\,$^4,3$ \and
Yannick Guesnet\,$^3$ \and
Thierry Lecroq\,$^3$ \and
Arnaud Lefebvre\,$^3$ \and
Martine Léonard\,$^3$ \and
Laurent Mouchard\,$^3$ \and
\'Elise Prieur-Gaston\,$^3$ \and
Bruce Watson\,$^{5,6}$}
\institute{Department of Computer Science, Aberystwyth Univ. (Mauritius Branch Campus), Quartier Militaire, Mauritius \and
Department of Informatics, King’s College London, UK \and
Normandie Univ., UNIROUEN, LITIS, 76000 Rouen, France \and
Modélisation, Information et Systèmes (MIS), Univ. de Picardie Jules Verne, Amiens, France \and
Department of Information Science, Stellenbosch Univ.,  South Africa \and
CAIR, CSIR Meraka, Pretoria, South Africa}
\authorrunning{Daykin \textit{et al}}
\maketitle

\begin{abstract}
A {\it degenerate} or {\it indeterminate} string on an alphabet
 $\alphabet$ is a sequence of non-empty subsets of $\alphabet$.
Given a degenerate string \s{t} of length $n$, 
 we present a new method based on the Burrows--Wheeler transform
 for searching for a degenerate
 pattern of length $m$ in \s{t} running in $O(mn)$ time
 on a constant size alphabet  $\alphabet$.
Furthermore, it is a {\it hybrid} pattern-matching technique that works on both
 regular and degenerate strings.
A degenerate string is said to be {\it conservative} if its number of non-solid
 letters is upper-bounded by a fixed positive constant $q$; in this case we
 show that the search complexity time is $O(qm^2)$.
Experimental results show that our method performs well in practice.
\end{abstract}


\section{Introduction}
\label{sec-intro}

An {\it indeterminate} or {\it degenerate} string $\s{x} = \s{x}[1 \pp n]$ on an
 alphabet $\alphabet$ is a sequence of non-empty subsets of $\alphabet$.
Degenerate strings date back to the groundbreaking
 paper of Fischer \& Paterson~\cite{FP74}.
This simple generalization of a regular string, from letters to subsets of
 letters, arises naturally in diverse applications: in musicology, for instance
 the problem of finding chords that match with single notes; search tasks
 allowing for occurrence of errors such as with web interfaces and search
 engines; bioinformatics activities including DNA sequence analysis and coding
 amino acids; and cryptanalysis applications.

For regular or solid strings, the main approaches for computing all the occurrences of a
 given nonempty pattern $\s{p} = \s{p}[1\pp m]$ in a given nonempty text
 $\s{t} = \s{t}[1\pp n]$ have been window-shifting techniques, and applying the
 bit-parallel processing to achieve fast processing -- for expositions of
 classic string matching algorithms see~\cite{CL04}.
More recently the Burrows--Wheeler transform (BWT) has been tuned to this search
 task, where all the occurrences of the pattern $\s{p}$ can be found as a prefix
 of consecutive rows of the BWT matrix, and these rows are determined using a
 backward search process.
 
The degenerate pattern matching problem for degenerate strings $\s{p}$ and
 $\s{t}$ over $\Sigma$ of length $m$ and $n$ respectively is the task of finding
 all the positions of all the occurrences of $\s{p}$ in $\s{t}$, that is,
 computing every $j$ such that $\forall\,1\le i \le |\s{p}|$ it holds that
 $\s{p}[i] \cap \s{t}[i+j]\ne \emptyset$.
  
Variants of degenerate pattern matching have recently been proposed.
A degenerate string is said to be {\it conservative} if its number of non-solid
 letters is upper-bounded by a fixed positive constant $q$.
Crochemore {\it et al.}~\cite{CIKMV16} considered the matching problem of
 conservative degenerate strings and presented an efficient algorithm that can
 find, for given degenerate strings $\s{p}$ and $\s{t}$ of total length $n$
 containing $q$ non-solid letters in total, the occurrences of $\s{p}$ in
 $\s{t}$ in $O(nq)$ time, i.e. linear in the size of the input.

Our novel contribution is to implement degenerate pattern matching by modifying
 the existing Burrows--Wheeler pattern matching technique.
Given a degenerate string $\s{t}$ of length $n$, searching for either a degenerate
 or solid pattern of length $m$ in $\s{t}$ is achieved in $O(mn)$ time; in the
 conservative scenario with at most $q$ degenerate letters, the search complexity
 is $O(qm^2)$ -- competitive for short patterns.
This formalizes and extends the work implemented in BWBBLE~\cite{bwbble}.

\section{Notation and definitions}
\label{subsec-notation}

Consider a finite totally ordered alphabet $\alphabet$ of constant size
 which consists of a set of \emph{letters}.
The order on letters is denoted by the usual symbol $<$.
A {\it string} is a sequence of zero or more letters over $\alphabet$.
The set of all strings over $\alphabet$ is denoted by $\alphabet^*$ and
 the set of all non-empty strings over $\alphabet$ is denoted by $\alphabet^+$.
Note we write strings in mathbold such as \s{x}, \s{y}.
The lexicographic order (\emph{lexorder}) on strings is also denoted by the symbol $<$.

A string $\s{x}$ over $\Sigma^+$ of length $|\s{x}|=n$ is represented by $\s{x}[1\pp n]$, where
 $\s{x}[i] \in \alphabet $ for $ 1\leq i \leq n $ is 
 the $i$-th letter of $\s{x}$.
The symbol $\sharp$ gives the number of elements in a specified set.

The concatenation of two strings $\s{x}$ and $\s{y}$ is defined as the sequence of letters of $\s{x}$ followed by
 the sequence of letters of $\s{y}$ and is denoted by $\s{x}\cdot \s{y}$ or simply $\s{x}\s{y}$ when no confusion is
 possible.
 A string $\s{y}$ is a {\it substring} of $\s{x}$ if $\s{x} = \s{u}\s{y}\s{v}$, where $\s{u},\s{v} \in \alphabet^{\ast}$; 
specifically a string $\s{y} = \s{y}[1 \pp m]$ is a substring of $\s{x}$ 
if $\s{y}[1 \pp m] = \s{x}[i \pp i+m-1]$ for some $i$.
Strings $\s{u}=\s{x}[1\pp i]$ are called {\it prefixes} of $\s{x}$, and strings $\s{v}=\s{x}[i \pp n]$ are called {\it suffixes} of $\s{x}$ of length $n$
 for $1\le i \le n$.
The prefix $\s{u}$ (respectively suffix $\s{v}$) is a proper prefix (suffix) of a string $\s{x}$ if $\s{x}\not=\s{u},\s{v}$.
A string $\s{y}=\s{y}[1\pp n]$ is a {\it cyclic rotation} of $\s{x}=\s{x}[1 \pp n]$ 
 if $\s{y}[1 \pp n] = \s{x}[i \pp n]\s{x}[1\pp i-1] $ for some $1 \leq i \leq n$.
 

\begin{definition}[Burrows--Wheeler transform]
The BWT of \s{x} is defined as the pair $(L,h)$ where
 $L$ is the last column of the matrix $M_{\s{x}}$ formed by all the lexorder sorted
 cyclic rotations of $\s{x}$ and  $h$ is the index of \s{x} in this matrix.
\end{definition}

The BWT is easily invertible via a linear last first mapping~\cite{BW94} using
 an array $C$ indexed by all the letters $c$ of the alphabet $\Sigma$ and 
 defined by:
$C[c] = \sharp \{i \mid \s{x}[i] < c\}$
and 
$\rank_c(\s{x},i)$ which gives the number of occurrences of the letter $c$ in the prefix $\s{x}[1\pp i]$.

In~\cite{DW17}, Daykin and Watson present a simple
 modification of the classic BWT,
 the {\it degenerate Burrows--Wheeler transform}, which is suitable for
 clustering degenerate strings.

Given an alphabet $\alphabet$ we define a new alphabet $\Deltasigma$
 as the non-empty subsets of $\alphabet$:
 $\Deltasigma={\cal P}(\alphabet) \setminus \{\emptyset\}$.

Formally a non-empty {\it indeterminate} or {\it degenerate} string
$\s{x}$ is an element of $\Deltasigma^+$.
We extend the notion of prefix on degenerate strings as follows.
A degenerate string $\s{u}$ is called a {\it degenerate prefix} of $\s{x}$ if  $|\s{u}| \leq  |\s{x} |$ and $\s{u}[i] \cap \s{x}[i] \neq \emptyset$ $\forall 1 \leq i \leq |\s{u}|$.

A degenerate string is said to be {\it conservative} if its number of
 non-solid letters is upper-bounded by a fixed positive constant $q$.

\begin{definition} 
\label{def-ind_cyclicrotation}
A degenerate string $\s{y} = \s{y}[1\pp n]$ 
is a \emph{degenerate cyclic rotation} of a degenerate string $\s{x} = \s{x}[1\pp n]$ 
if $\s{y}[1 \pp n] = \s{x}[i \pp n]\s{x}[1\pp i-1] $ for some $1 \leq i \leq n$ (for $i =1, \s{y} =\s{x}$).
\end{definition} 

Given an order on $\Deltasigma$ denoted by the usual symbol $<$,
 we can compute the BWT of a degenerate string $\s{x}$ in the same way as for a regular string; here we apply lexorder.

\section{Searching for a degenerate pattern in a degenerate string}

Let $\s{p}$ and $\s{t}$ be two degenerate strings over $\Deltasigma$ of length $m$ and $n$ respectively.
We want to find the positions of all the occurrences or matches of $\s{p}$ in $\s{t}$ i.e. we want to compute
 every
 $j$ such that $\forall\,1\le i \le |\s{p}|$ it holds that $\s{p}[i] \cap \s{t}[i+j]\ne \emptyset$.
For determining the matching we will apply the usual backward search
 but at each step we may generate several different intervals which will be stored in a set $\H$.
Then step $k$ (processing $\s{p}[k]$ with $1\le k\le m$) of the backward search can be formalized as follows:
$$
\begin{array}{lll}
\Step(\H,k,C,\BWT=(L,h),\s{p})=
(((r,s)) &\mid & r=C[c]+\rank_c(L,i-1)+1,\cr
         &     & s=C[c]+\rank_c(L,j),\cr
&& r \le s, \, (i,j)\in \H,\, c\in\Deltasigma \mbox{ and } c \cap \s{p}[k] \ne\emptyset).
\end{array}$$

Let $\Stepi(m,C,\BWT,\s{p})=\Step(\{(1,n)\}, m,C,\BWT,\s{p})$ and
    $\Stepi(i,C,\BWT,\s{p})=\Step(\Stepi(i+1,C,\BWT,\s{p}),i,C,\BWT,\s{p})$ for $1\le i \le m-1$.
In words, $\Stepi(i,C,\BWT,\s{p})$ applies step $m$ through to $i$ of the backward search.

\begin{lemma}
\label{lemm-stepi}
The interval $(i,j)\in \Stepi(k,C,\BWT,\s{p})$ if and only if $\s{p}[k\pp m]$
 is a degenerate prefix of $M_{\s{t}}[h]$ for $i\le h\le j$.
\end{lemma}

\begin{proof}
$\Longrightarrow$:
By induction.
By definition of the array $C$ when $(i,j)\in\Stepi(m,C,\BWT,\s{p})$ then
 $\s{p}[m]$ is a degenerate prefix of $M_{\s{t}}[h]$, for $i\le h\le j$.
So assume that the property is true for all the integers $k'$ such that $k< k' \le m$.
If $(r,s)\in \Stepi(k,C,\BWT,\s{p})$ then $r=C[a]+\rank_a(\BWT,i-1)+1$ and $s=C[a]+\rank_a(\BWT,j)$
 with $r \le s$, $(i,j)\in \Stepi(k+1,C,\BWT,\s{p})$, $a\in\Deltasigma$ and $a \cap \s{p}[k] \ne\emptyset$.
Thus by the definition of the BWT, $\s{p}[k\pp m]$ is a degenerate prefix of rows of $M_{\s{t}}[h]$ for $r\le h\le s$.

$\Longleftarrow$:
By induction.
By definition, if $\s{p}[m]$ is a degenerate prefix of $M_{\s{t}}[h]$ for $r\le h\le s$ then $(r,s)\in \Stepi(m,C,\BWT,\s{p})$.
So assume that the property is true for all integers $k'+1$ such that $k<k' \le m$.
If $\s{p}[k+1\pp m]$ is a degenerate prefix of $M_{\s{t}}[h]$ for $i\le h\le j$ then
 $(i,j)\in \Stepi(k+1,C,\BWT,\s{p})$.
 When $\s{p}[k\pp m]$ is a degenerate prefix of $M_{\s{t}}[h]$ for $r\le h\le s$ then
 $(r,s) \in \Step(\Stepi(k+1,C,\BWT,\s{p}),i,C,\BWT,\s{p}) = \Stepi(k,C,\BWT,\s{p})$
 by definition of the array $C$ and of the rank function.

We conclude that the property holds for $1\le k \le m$.
\end{proof}

\begin{corollary}
\label{coro-stepi}
The interval $(i,j)\in \Stepi(1,C,\BWT,\s{p})$ if and only if $\s{p}$ is a degenerate prefix of
 $M_{\s{t}}[h]$ for $i\le h\le j$.
\end{corollary}

The proposed algorithm, see \figurename~\ref{figu-algo-dbs}, computes
 $\Stepi(1,C,\BWT,\s{p})$ by first initializing the variable $\H$ with $\{(1,n)\}$
 and then performing steps $m$ to $1$, while
 exiting whenever $\H$ becomes empty.

\begin{figure}[t]
\begin{algo}{DegenerateBackwardSearch}{\s{p},m,\BWT=(L,h),n,C}
  \SET{\H}{\{(1,n)\}}
  \SET{k}{m}
  \label{outerloop-begin}
  \DOWHILE{\H \ne \emptyset \mbox{ and } k \ge 1}
    \SET{\H'}{\emptyset}
    \DOFOR{ (i,j) \in \H}
      \DOFOR{c\in\Deltasigma \mbox{ such that } c \cap \s{p}[k]\ne\emptyset}
        \SET{\H'}{\H' \cup \{(C[c]+\rank_c(L,i-1)+1, C[c]+\rank_c(L,j))\}}
      \OD
    \OD
    \SET{\H}{\H'}
    \label{outerloop-end}
    \SET{k}{k-1}
  \OD 
  \RETURN{\H}
\end{algo}
\caption{
\label{figu-algo-dbs}
Backward search for a degenerate pattern in the BWT of a degenerate string. 
}
\end{figure}

The following two lemmas show that the number of intervals in $\H$ cannot grow exponentially.

\begin{lemma}
\label{lemm-overlap}
The intervals in $\Step(\{(i,j)\},k,C,\BWT,\s{p})$ do not overlap.
\end{lemma}

\begin{proof}
$\Step(\{(i,j)\},k,C,\BWT,\s{p})$ will generate one interval for every distinct letter $c\in\Deltasigma$ 
 such that $c \cap \s{p}[k] \ne\emptyset$.
Thus these intervals cannot overlap.
\end{proof}

\begin{lemma}
\label{lemm-overlaps}
The intervals in $\Step(\{(i,j),(i',j')\},k,C,\BWT,\s{p})$ with $i \leq j < i' \leq j'$ do not overlap.
\end{lemma}

\begin{proof}
From Lemma \ref{lemm-overlap}, the intervals generated from $(i,j)$ do not overlap, and the intervals generated from $(i',j')$ do not overlap.

Let $(r,s)$ be an interval generated from $(i,j)$, and
let $(r',s')$ be an interval generated from $(i',j')$.
Formally, let $r,s,c$ be such that 
$ r=C[c]+\rank_c(\BWT,i-1)+1$,
$ s=C[c]+\rank_c(\BWT,j)$,
$ c\in\Deltasigma \mbox{ and } c \cap \s{p}[k] \ne\emptyset$.
Let $r',s',c'$ be such that 
$ r'=C[c']+\rank_{c'}(\BWT,i'-1)+1$,
$ s'=C[c']+\rank_{c'}(\BWT,j')$,
$ c'\in\Deltasigma \mbox{ and } c' \cap \s{p}[k] \ne\emptyset$.

If $c\ne c'$ then $(r, s)$ and 
 $(r', s')$ cannot overlap
 since $C[c] \leq r \leq s < C[c]+\sharp\{ i \mid \s{t}[i] = c\}$ and $C[c'] \leq r' \leq s' < C[c']+\sharp\{ i \mid \s{t}[i] = c'\}$.
Otherwise if $c=c'$ then
 since $j<i'$,
it follows that $\rank_c(\BWT,j) < \rank_c(\BWT,i'-1)+1$
  and thus $(r,s)=(C[c]+\rank_c(\BWT,i-1)+1, C[c]+\rank_c(\BWT,j))$ and 
 $(r',s')=(C[c]+\rank_c(\BWT,i'-1)+1, C[c]+\rank_c(\BWT,j'))$ do not overlap.
\end{proof}

\begin{corollary}
\label{cor-overlap}
Let $H$ be a set of non-overlapping intervals.
The intervals in\\$\Step(H,k,C,\BWT,\s{p})$ do not overlap.
\end{corollary}

We can now state the complexity of the degenerate backward search.

\begin{theorem}
\label{theo-dbs}
The algorithm \CALL{DegenerateBackwardSearch}{p,m,\BWT,n,C} computes a set of intervals
 $\H$, where $(i,j)\in\H$ if and only if $\s{p}$ is a degenerate prefix of consecutive rows of $M_{\s{t}}[k]$ for $i\le k \le j$,
 in time $O(mn)$ for a constant size alphabet.
\end{theorem}

\begin{proof}
The correctness comes from Corollary~\ref{coro-stepi}.
The time complexity mainly comes from Lemma~\ref{lemm-overlap} and the fact that
 the alphabet size is constant.
\end{proof}

For conservative degenerate string the overall complexity of the search can be reduced.

\begin{theorem}
\label{theo-cons}
Let \s{t} be a conservative degenerate string over a constant size
 alphabet.
Let the number of degenerate letters of \s{t} be bounded by a constant $q$.
Then given the BWT of \s{t},
 all the intervals in the BWT of occurrences
 of a pattern \s{p} of length $m$ can be detected in time $O(qm^2)$.
\end{theorem}

\begin{proof}
The number of intervals of occurrences of \s{p} in \s{t} that do not overlap
 a degenerate letter is bounded by $k+1$.
The number of intervals of occurrences of \s{p} in \s{t} that overlap one degenerate
letter is bounded by $k$ times $m$.
Since there are at most $k$ degenerate letters in \s{t} and since
 the backward search has $m$ steps the result follows.
\end{proof}

From Corollary \ref{cor-overlap},  the number of intervals at each step of the backward search cannot exceed $n$.
However, in practice, it may be worthwhile decreasing the number of intervals further:
 the next lemma shows that adjacent intervals can be merged.
In order to easily identify adjacent intervals we will now store them in a sorted list-like data
 structure as follows.
For two lists $I$ and $J$ the concatenation of the elements of $I$ followed by
 the elements of $J$ is denoted by $I\cdot J$.

We proceed to define the operation $\Merge$ that consists in merging two adjacent intervals:
$\Merge(\emptyset) = \emptyset$
and
$\Merge((i,j)) = ((i,j))$,
$\Merge(((i,j),(j+1,j'))\cdot I)=\Merge(((i,j'))\cdot I)$,
$\Merge(((i,j),(i',j'))\cdot I)=((i,j))\cdot \Merge(((i',j'))\cdot I)$ for $i'>j+1$.
The next lemma justifies the merging of adjacent intervals in $\H$.
 
\begin{lemma}
\label{lemm-merge}
$\Merge(\Step(((i,j),(j+1,j')),k,C,\BWT,\s{p}))=$\\$\Merge(\Step(((i,j')),k,C,\BWT,\s{p}))$.
\end{lemma}

\begin{proof}
For a letter $c\in\Deltasigma$ such that $c\cap \s{p}[k] \ne \emptyset$ the intervals generated
 from $(i,j)$ and $(j+1,j')$ are, by definition, necessarily adjacent
 which shows that if $(p,q)\in\Merge(\Step(((i,j),(j+1,j')),k,C,\BWT,\s{p}))$ 
 then $(p,q)\in\Merge(\Step(((i,j')),k,C,\BWT,\s{p}))$.
The reciprocal can be shown similarly.
\end{proof}

This means that $\H$ can be implemented with an efficient data structure such as red-black trees
 adapted for storing non-overlapping
 and non-adjacent intervals.

\section{Experiments}
\label{sect-expe}

We ran algorithm \textsc{DegenerateBackwardSearch} (DBS) for searching for the
 occurrences
 of a degenerate pattern in different random strings: solid strings, degenerate strings
 and conservative degenerate strings.
The alphabet consists of subsets of the DNA alphabet encoded by integers
 from 1 to 15.
Solid letters are encoded by powers of 2 (1, 2, 4 and 8) as in~\cite{bwbble}.
Then intersections between degenerate letters can be performed by a bitwise \texttt{and}.
The conservative string contains $500,000$ degenerate letters.

We also ran the adaptive Hybrid pattern-matching algorithm of~\cite{SWY08},
 and, since the alphabet size is small
 we also ran a version of the Backward-Non-Deterministic-Matching (BNDM)
 adapted for degenerate pattern matching (see \cite{NR2002}).
The Hybrid and BNDM are bit-parallel algorithms and have only been tested for pattern
 lengths up to 64.
The patterns have also been randomly generated.
For the computation of the BWT we used the SAIS library~\cite{sais} and for its implementation
 we used the SDSL library~\cite{gbmp2014sea}.
All the experiments have been performed on a computer with a 1.3 GHz Intel Core
 i5 processor and 4 GB 1600 MHz DDR3 RAM.

We performed various experiments and present only two of them.
For DBS the measured times exclude the construction of the BWT
 but include the reporting of the occurrences using a suffix array.
This can be justified by the fact that, in most cases, strings are given in a compressed form
 through their BWTs.
\figurename~\ref{figu-data3}(a) shows the searching times
 for different numbers of degenerate patterns of length $8$ in a solid string.
Times are in centiseconds.
It can be seen that when enough patterns have to be searched for
 in the same string then it is worth using the new DBS algorithm.
The BNDM algorithm performs better than the Hybrid one due to the small
 size of the alphabet which favors shifts based on suffixes of the pattern rather than
 shifts based on single letters.
\figurename~\ref{figu-data3}(b) shows the searching times
 for a degenerate pattern of length $8$ in conservative degenerate strings
 of various lengths
 (for each length the strings contain $10\%$ of degenerate letters).
As expected when the length of the string increases the advantage of using DBS
 also increases.

\begin{figure}[t]
\begin{center}
\setlength{\tabcolsep}{0pt}
\begin{tabular}{cc}
\includegraphics[scale=0.24,angle=-90]{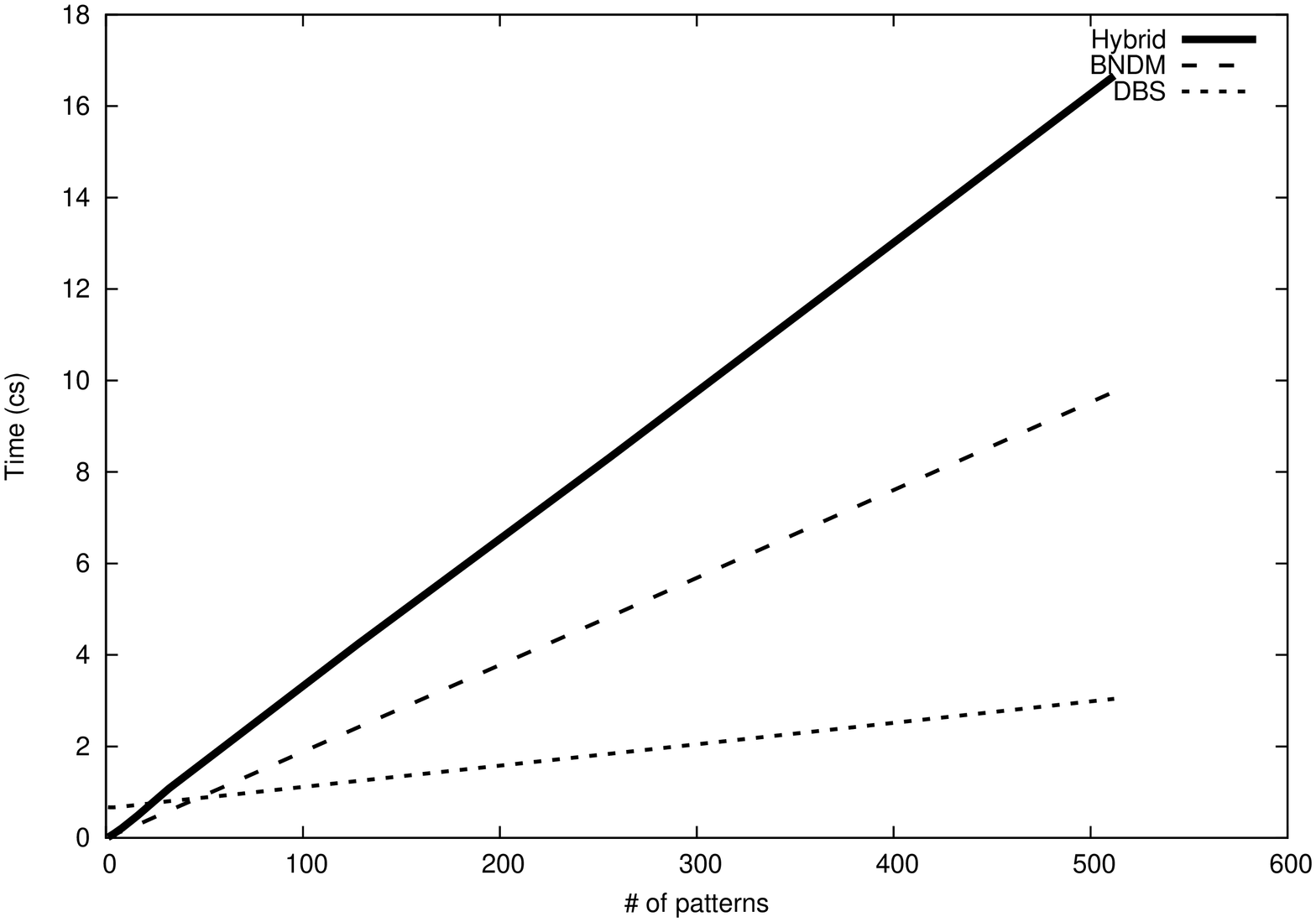}&\includegraphics[scale=0.24,angle=-90]{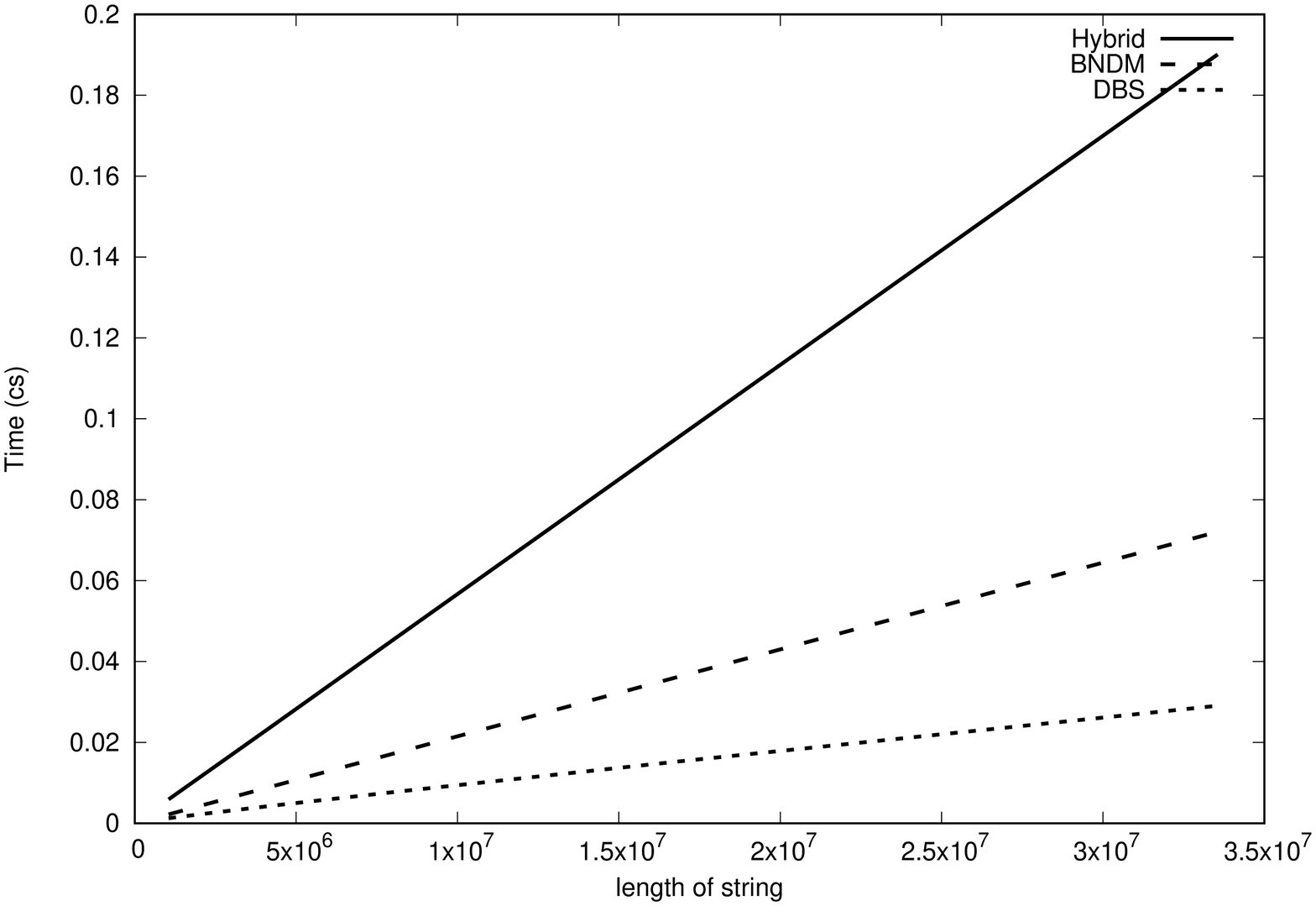}\\
(a)&(b)
\end{tabular}
\end{center}
\caption{
\label{figu-data3}
(a): Running times for searching for several degenerate patterns of length $8$ in a solid string of length $5$MB.
(b): Running times for searching for one degenerate pattern of length 8 in 
 a conservative degenerate string of variable length.
}
\end{figure}

%



\bibliographystyle{splncs03}
\bibliography{ref}


\end{document}